\documentclass[a4paper]{article}
\usepackage{array}
\usepackage{amsthm}
\usepackage{graphicx}
\usepackage{float}
\usepackage{amsfonts}
\usepackage{amsmath}
\usepackage{amssymb}
\usepackage{color}

\newcommand {\reals} {\ensuremath{\mathbb{R}}}

\usepackage{amsthm}
\usepackage{amssymb}
\usepackage{amsmath}
\usepackage{enumerate}
\usepackage{textcomp}
\usepackage{verbatim}
\newtheorem{theorem}{Theorem}
\newtheorem{lemma}[theorem]{Lemma}

\newcommand{\degree}
{\ensuremath{^\circ}}

\usepackage{graphicx}
\usepackage{grffile}
\title{Selection Lemmas for various geometric objects}
\author{Pradeesha Ashok \thanks{Department of Computer Science and Automation,
				Indian Institute of Science, Bangalore, India.
				Email :\texttt{pradeesha@csa.iisc.ernet.in}}
\and Ninad Rajgopal\thanks{Department of Computer Science and Automation,
				Indian Institute of Science, Bangalore, India.
				Email :\texttt{ninad.rajgopal@csa.iisc.ernet.in}}
\and Sathish Govindarajan\thanks{Department of Computer Science and Automation,
				Indian Institute of Science, Bangalore, India.
				Email :\texttt{gsat@csa.iisc.ernet.in}}}
\begin{document}
\maketitle
%
%
%
\begin{abstract}
Selection lemmas are classical results in discrete geometry that have been well studied and have
applications in many geometric problems like weak epsilon nets and slimming Delaunay triangulations. Selection lemma 
type results typically show that there exists a point that is contained in many objects that are induced (spanned) by an underlying point set. 

In the first selection lemma, we consider the set of all the objects induced (spanned) by a point set $P$. This question has been widely 
explored for simplices in $\reals^d$, with tight bounds in $\reals^2$. In our paper, we prove first selection lemma for other 
classes of geometric objects. We also consider the strong variant of this problem where we add the constraint that the piercing point 
comes from $P$. We prove an exact result on the strong and the weak variant of the first selection lemma for 
axis-parallel rectangles, special subclasses of axis-parallel rectangles like quadrants and slabs, disks (for centrally symmetric point sets). 
We also show non-trivial bounds on the first selection lemma for axis-parallel boxes and hyperspheres in $\mathbb{R}^d$.

In the second selection lemma, we consider an arbitrary $m$ sized subset of the set of all objects induced by $P$. We study this problem for axis-parallel rectangles and show that there exists an point in the plane that is contained in $\frac{m^3}{24n^4}$ rectangles. This is an improvement over the previous bound by 
Smorodinsky and Sharir~\cite{SHA04} when $m$ is almost quadratic.
\end{abstract}

\section{Introduction}


Let $P$ be a set of points in $\reals^d$. Consider the family of all objects $\mathcal{R}$ of a particular kind
(eg. hyperspheres, boxes, simplices, \dots)  such that each object in $\mathcal{R}$ has a distinct tuple of points from $P$ on
its boundary. For example, in $\reals^2$, $\mathcal{R}$ could be the family of $n \choose 3$ triangles such that each triangle has a
distinct triple of points of $P$ as its vertices. $\mathcal{R}$ is called the set of all objects induced (spanned) by $P$.
Various questions related to geometric objects induced by a point set have been studied in the last few decades. In this 
paper, we look at the problem of bounding the largest subset of $\mathcal{R}$ that can be hit/pierced by a single point.

Combinatorial results on these questions are referred as {\em Selection Lemmas} and are well studied. A classical result in 
discrete geometry is the {\em First Selection Lemma}~\cite{BF84} which shows that there exists a point that is present in 
$\frac{2}{9} \cdot \binom{n}{3}$ (constant fraction of) triangles induced by $P$. Bukh~\cite{BUK06} provides a simple and 
elegant proof of the above statement. Moreover, it is known that the constant in this result is tight~\cite{BUK10}. Interestingly, both~\cite{BF84, BUK10} use the centerpoint as the piercing point.

Let $P$ be a set of $n$ points in $\mathbb{R}^d$. A point $x \in \mathbb{R}^d$ is said to be a \emph{centerpoint} of $P$ if any halfspace that contains $x$ contains at least $n\over{d+1}$ points of $P$. Equivalently, $x$ is a centerpoint if and only if $x$ is contained in every convex object that contains more than $\frac{d}{d+1} n$ points of $P$. It has been proved that a centerpoint exists for any point set $P$ and the constant $\frac{d}{d+1}$ is tight~\cite{Rad46}. The centerpoint question has also been studied for special classes of convex objects like axis-parallel rectangles, halfplanes and disks~\cite{AAH09}. Another variant of the centerpoint called strong centerpoint, where the centerpoint is required to be an input point, has also been studied~\cite{AGK10}.

The first selection lemma has also been considered for simplices in $\reals^d$. This is an important result in discrete geometry 
and it has been used in the construction of weak $\epsilon$-nets for convex objects~\cite{MAT02}. B{\'a}r{\'a}ny~\cite{BAR82} 
showed that there exists a point $p \in \reals^d$ contained in at least $c_d \cdot \binom{n}{d+1} - O(n^d)$ simplices induced 
from $P$, where $c_d \geq \frac{1}{(d+1)^d}$. Wagner~\cite{WAG03} improved this bound to $c_d \geq \frac{d^2+1}{(d+1)^{d+1}}$. 
Gromov~\cite{GRO10} developed a new topological method which established an improved lower bound of $c_d \geq 
\frac{2d}{(d+1)!(d+1)}$. Furthermore, Karasev~\cite{KAR12} gave a simplified and elegant proof for Gromov's bound and Matousek 
et al.~\cite{MAT11} provided an exposition of the combinatorial components in Gromov's proof. For the upper bound, Bukh et 
al.~\cite{BUK10} showed that there exists a point set in $\reals^d$ such that no point is present in more than 
$(\frac{n}{d+1})^{d+1} + O(n^d)$ induced simplices i.e. $c_d \leq \frac{(d+1)!}{(d+1)^{(d+1)}}$. For $d=2$, this shows that 
the bound for $c_d$ is tight. Furthermore they conjectured that this bound was tight for $d \geq 3$. For the case of 
$\reals^3$, Basit et al.~\cite{BAS10} improved the lower bound for the first selection lemma in $\reals^3$ and showed that there exists a point present in $0.00227 \cdot 
n^4$ simplices (tetrahedrons) spanned by $P$ i.e. $c_3 \geq 0.05448$. Further improvements on $c_3$ were shown in~\cite{GRO10,KRAL12,MAT11}, with $c_3 \geq 0.07480$ being the best known lower bound~\cite{KRAL12}. 

A generalization of the first selection lemma, known as the {\em Second Selection Lemma}, considers an $m$-sized arbitrary subset
$\mathcal{S} \subseteq \mathcal{R}$ of distinct induced objects of a particular kind and shows that there exists a point which
is contained in $f(m,n)$ objects of $\mathcal{S}$. The second selection lemma has been considered for various objects like 
simplices, boxes and hyperspheres in $\reals^d$ ~\cite{ABFK92,ARO90,CHA94,SHA04}. Aronov et al.~\cite{ARO90} showed that for any 
set $P$ of $n$ points and any set $T$ of $t$ triangles induced by $P$, there exists a point $p$ in the interior of at least 
$f(t,n) = \frac{t^3}{2^9n^6\log^5n}$, when $t=n^{3-\alpha}, \alpha \leq 1$. Their motivation was to derive an upper bound on the 
number of halving planes of a finite set of points in $\reals^3$. Alon et al.~\cite{ABFK92} showed that, for any family $F$ of $\alpha \binom{n}{d+1}$ 
induced simplices, there exists a point contained in at least $c\alpha^{s_d} \binom{n}{d+1}$ simplices 
of $F$, where $c, s_d$ are constants.

Chazelle et al.~\cite{CHA94} looked at this problem for hyperspheres with the motivation of reducing the complexity of Delaunay 
triangulations for points in $\reals^3$. 
They proved a selection lemma for intervals in the line and then extended it for axis-parallel boxes in $\mathbb{R}^d$, by induction on dimension. This in turn was used for the proof of the selection lemma for 
diametrical spheres induced by a pair of points, by using the fact that any diametrical sphere induced by a pair of points would contain the corresponding induced axis-parallel box. This gave a bound of 
$\Omega\left(\frac{m^2}{n^2\log^{2d-2}(\frac{n^2}{m})}\right)$ for rectangular boxes in $d$ dimensions (and hence the diametrical 
hyperspheres as well) and was extended to $\Omega\left(\frac{m^2}{n^2\log^{2d}(\frac{n^2}{m})}\right)$ for general hyperspheres in $d$ 
dimensions.

Smorodinsky and Sharir~\cite{SHA04} improved the bounds obtained in \cite{CHA94} by using a probabilistic proof very similar
to the one used in the proof of Crossing lemma \cite{MAT02}. Note that this paper proved that the point which pierced a lot of
disks (pseudo-disks) and the $d$-dimensional hyperspheres came from $P$. In the case of the axis-parallel rectangles, 
they proved a lower bound of $\Omega(\frac{m^2}{n^2\log^2n})$ and an improved upper bound of 
$O(\frac{m^2}{n^2\log(\frac{n^2}{m})})$. However, in this case the piercing point could be any point in $\reals^2$.

As mentioned earlier, first selection lemma has been extensively studied for simplices in $\mathbb{R}^d$. However, no previous work is known on first selection
lemma for other geometric objects, to the best of our knowledge. In our paper, we explore the first selection lemma for other geometric objects like axis-parallel boxes and hyperspheres in $\mathbb{R}^d$. We call the case where the piercing point $p \in 
\reals^d$ (same as the previous literature) as the {\em weak variant}. We also consider the {\em strong variant} of the first selection lemma where we add 
the constraint that the piercing point $p \in P$. We prove an exact result on the strong and weak variant of the first selection lemma for axis-parallel rectangles, quadrants, slabs and disks (for centrally symmetric point sets). Note that the first selection lemma for triangles~\cite{BF84,BUK10} used the centerpoint as the piercing point to prove exact bounds. Interestingly, we also use the strong and weak centerpoint for the respective objects to prove our results in 
sections~\ref{sec_rec},~\ref{spl_rec} and ~\ref{disks}.

Let $P$ be a set of $n$ points in $\reals^d$ in general position i.e., no two points have the same coordinate in any dimension 
and no $d+2$ points lie on the same hypersphere. Let $\mathcal{F}$ be a family of objects induced by $P$. For any point $p$, 
let $\mathcal{F}_p \subseteq \mathcal{F}$ be the set of objects that contain $p$ and $f_p^\mathcal{F} = \vert \mathcal{F}_p \vert$. Let $s^\mathcal{F}(n)$ and $w^\mathcal{F}(n)$ denote the bounds for the strong and the weak 
variant of the first selection lemma for a family of objects $\mathcal{F}$. In particular,

\begin{equation*}
\begin{split}
s^\mathcal{F}(n) & = \displaystyle\min_{P,|P|=n} (\max_{p \in P} f_p^\mathcal{F}) \\
w^\mathcal{F}(n) & = \displaystyle\min_{P,|P|=n} (\max_{p \in \reals^d} f_p^\mathcal{F}) \\
\end{split}
\end{equation*}

Our results for the first selection lemma for various families of objects are summarized in Table \ref{results}.

\begin{table}[H]
\centering
\begin{tabular}{|c|c|c|c|c|c|}
\hline
\textbf{Family of Objects $\mathcal{F}$} &\textbf{Dim}  &\multicolumn{2}{c|}{\textbf{$s^\mathcal{F}(n)$}} 
                                                                &\multicolumn{2}{c|}{\textbf{$w^\mathcal{F}(n)$}} \\
\hline
& &Lower &Upper &Lower &Upper \\
& &Bound &Bound &Bound &Bound \\
\hline
Axis-parallel rectangles &2 &\multicolumn{2}{c|}{$n^2/16$} &\multicolumn{2}{c|}{$n^2/8$} \\
\hline
Axis-parallel boxes      &$d$ &\multicolumn{2}{c|}{-}        &$\frac{n^2}{2^{(2^d-1)}}$ &$\frac{n^2}{2^{(d+1)}}$ \\
\hline
Orthants                 &2 &\multicolumn{2}{c|}{$n^2/4$}  &\multicolumn{2}{c|}{$n^2/2$} \\
\hline
Axis-parallel slabs      &2 &\multicolumn{2}{c|}{$3n^2/8$} &\multicolumn{2}{c|}{$n^2/2$} \\
\hline
Skylines                 &2 &$n^2/9$  &$n^2/8$              &\multicolumn{2}{c|}{$n^2/4$} \\
\hline
Disks                    &2 &$n^2/16$ &$n^2/9$              &$n^2/6$ &$n^2/4$ \\

\hline
Disks (Centrally         &   &\multicolumn{2}{c|}{}        &\multicolumn{2}{c|}{} \\
Symmetric Point Sets)    &2  &\multicolumn{2}{c|}{$n^2/8$} &\multicolumn{2}{c|}{$n^2/4$} \\
\hline
Hyperspheres             &$d$ &\multicolumn{2}{c|}{-}        &$\frac{n^2}{2(d+1)}$ &$n^2/4$ \\

\hline
Hyperspheres (Centrally  &    &\multicolumn{2}{c|}{}        &\multicolumn{2}{c|}{} \\
Symmetric Point Sets)    &$d$ &\multicolumn{2}{c|}{-}        &\multicolumn{2}{c|}{$n^2/4$} \\
\hline
\end{tabular}
\caption{First selection lemma Bounds for various families of objects}
\label{results}
\end{table}

We next consider the second selection lemma for axis-parallel rectangles in $\reals^2$. We prove that there exists a point $p \in \reals^2$ 
that is contained in at least $\frac{m^3}{24n^4}$ axis-parallel rectangles of $\mathcal{S}$. This bound is an improvement over the previous 
bound in \cite{CHA94,SHA04} when $m = \Omega (\frac{n^2}{\log^2 n})$. We use an elegant double counting argument to obtain this result. 






In section \ref{sec_rec}, we prove exact results for strong and weak variants of first selection lemma for axis-parallel rectangles. 
Section \ref{spl_rec} proves tight or almost tight bounds for the strong and weak variants of first selection lemma for families 
of special rectangles like orthants, slabs and skylines. In section \ref{boxes}, we prove bounds for the weak variant of first 
selection lemma for boxes in $\mathbb{R}^d$. In section \ref{disks}, we prove bounds for the strong variant of first selection 
lemma for induced disks in $\mathbb{R}^2$ and prove bounds for the weak variant of first selection lemma for hyperspheres in 
$\mathbb{R}^d$. Section \ref{ssl} proves improved bounds for second selection lemma for axis-parallel rectangles.

\section{Rectangles}
\label{sec_rec}

In this section, we prove the first selection lemma for axis-parallel rectangles. Let $R(u,v)$ be the axis-parallel rectangle induced by $u$ 
and $v$ where $ u,v \in P$ i.e., $R(u,v)$ has $u$ and $v$ as diagonal points. Let $\mathcal{R}$ be the set of all induced axis-parallel 
rectangles $R(u,v)$ for all $u,v \in P$. Let $p$ be any point and  $v$ and $h$ be the vertical and horizontal lines passing through $p$, 
dividing the plane into four quadrants as shown in figure \ref{rectangle}. Let $\vert A \vert$ represent $\vert A \cap P \vert$ 
(similar for all quadrants). $\mathcal{R}_p$ consists of exactly those rectangles which are induced by a pair of points present in diagonally 
opposite quadrants.

\subsection{Weak variant}
In this section, we obtain tight bounds for $w^\mathcal{R}(n)$ .
\begin{theorem}
\label{thm:selectionRectanglesLower}
$w^\mathcal{R}(n) = \frac{n^2}{8}$.
\end{theorem}

\begin{proof}
Let $p$ be the weak centerpoint for rectangles \cite{AAH09}. We claim that $f^\mathcal{R}_p \geq \frac{n^2}{8}$.

\begin{figure}[h]
\begin{minipage}[h]{0.5\linewidth}
\centering
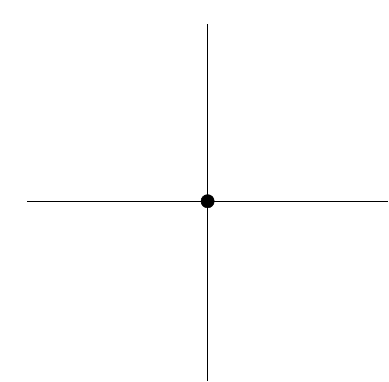
\caption {Lower bound}
\label {fig:lowerBound}
\end{minipage}
\hspace{1 pt}
\begin{minipage}[h]{0.5\linewidth}
\centering
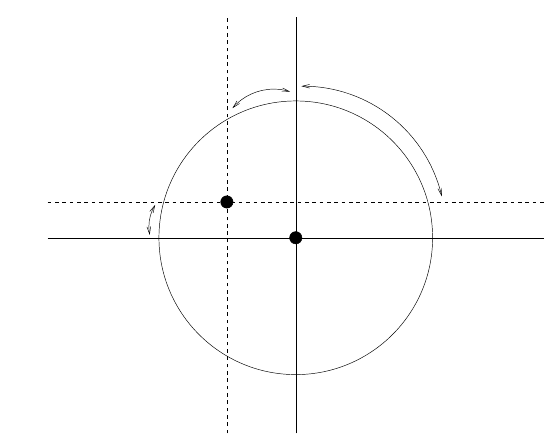
\caption {Upper bound construction}
\label {fig:upperBoundWeak}
\end{minipage}
\end{figure}

Let $p$ divide the plane into four quadrants as shown in figure \ref{fig:lowerBound}. W.l.o.g let the top left quadrant contain $(\frac{n}{4}+x)$ points. Therefore, the remaining 
points are distributed among the three other quadrants. Then,

\begin{eqnarray*}
f^\mathcal{R}_p  &=& \left(\frac{n}{4} - x\right)^2 + \left(\frac{n}{4} + x\right)^2  \\
                 &=& 2 \cdot \left(\frac{n^2}{16}\right) + 2 \cdot x^2 \\
\end{eqnarray*}

Thus, $f^\mathcal{R}_p \geq \frac{n^2}{8}$. Therefore, $w^\mathcal{R}(n) \geq \frac{n^2}{8}$. \\

For the upper bound, consider a set $P$ of $n$ points uniformly arranged along the boundary of a circle. Let $h$ and 
$v$ be horizontal and vertical lines that bisect $P$, intersecting at $o$. W.l.o.g, let $p$ be any point inside the circle in 
the top left quadrant and let $h_1$ and $v_1$ be the horizontal and vertical lines passing through $p$. Let $a$ be the number 
of points from $P$ below $h_1$ that is present in the top left quadrant defined by $h$ and $v$. Similarly, let $b$ be the 
number of points from $P$ to the right of $v_1$ that is present in the top left quadrant defined by $h$ and $v$. The number 
of points in each of the four quadrants defined by $h_1$ and $v_1$ is as shown in figure \ref{fig:upperBoundWeak}.
\begin{eqnarray*}
f^\mathcal{R}_p  &=& \left(\frac{n}{4}-b+a\right) \cdot \left(\frac{n}{4}-a+b\right) + \left(\frac{n}{4}-a-b\right) \cdot \left(\frac{n}{4}+a+b\right)  \\
                 &=& \frac{n^2}{8} - 2(a^2+b^2) \\
\end{eqnarray*}
Since $a,b \geq 0$, $f^\mathcal{R}_p \leq \frac{n^2}{8}$ for all points $p \in \mathbb{R}^2$. Therefore, $w^\mathcal{R}(n) \leq \frac{n^2}{8}$.
\end{proof}



\subsection{Strong variant}
In this section, we obtain exact bounds for $s^\mathcal{R}(n)$.

\begin{theorem}\label{strong_rec}
\label{strong_rect}
$s^\mathcal{R}(n) =\frac{n^2}{16}$
\end{theorem}

\begin{proof}

\begin{figure}[h]
\begin{minipage}[h]{0.45\linewidth}
\centering
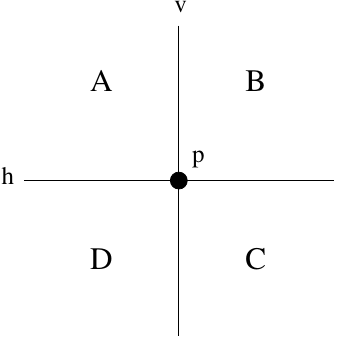
\caption{Lower bound}
\label{rectangle}
\end{minipage}
\hspace{1 pt}
\begin{minipage}[h]{0.45\linewidth}
\centering
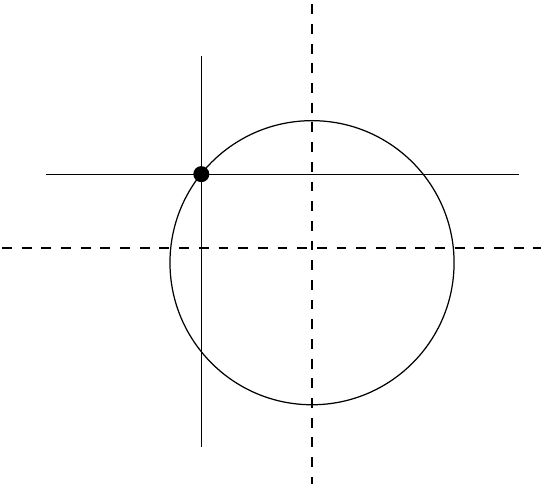
\caption {Upper bound construction}
\label {fig:upperBoundStrong}
\end{minipage}
\centering
\end{figure}

Let $p$ be the strong centerpoint of $P$ w.r.t axis-parallel rectangles. Then any axis-parallel rectangle that contains more 
than $\frac{3n}{4}$ points from $P$ contains $p$ \cite{AGK10}. We claim that $p$ is contained in at least $\frac{n^2}{16}$ 
rectangles from $\mathcal{R}$.  

Let $p$ partition $P$ into four quadrants as 
shown in figure \ref{rectangle}. If $\vert A \vert,
\vert C \vert \geq \frac{n}{4}$, then $p$ is contained in at least $\frac{n^2}{16}$ rectangles from $\mathcal{R}$. Therefore, assume 
$\vert A \vert=\frac{n}{4}-x$. Now, there are two cases.
\\

\noindent Case 1. $\vert C \vert \leq \frac{n}{4}$ : W.l.o.g, assume that $\vert C \vert =\frac{n}{4}-y$ and $x \geq y$. Therefore 
$\vert B \cup D \vert =\frac{n}{2}+x+y$. The value of $f^\mathcal{R}_p$ is minimized when the value of $\vert B \vert \times \vert D \vert$ 
is minimized. Since $\vert A \vert=\frac{n}{4}-x$ and there can be at most $\frac{3n}{4}$ points on either sides of $h$ and $v$, 
both $B$ and $D$ contain at least $x$ points. Therefore, $f^\mathcal{R}_p$ is minimized when $\vert B \vert =\frac{n}{2}+y$ and $\vert D \vert =x$. 
Then,
\begin{eqnarray*}
f^\mathcal{R}_p &\geq& \left(\frac{n}{4}-x\right)\left(\frac{n}{4}-y\right)+\left(\frac{n}{2}+y\right)x \\
      &\geq& \frac{n^2}{16}\\
\end{eqnarray*}

\noindent Case 2. $\vert C \vert > \frac{n}{4}$ : Assume $\vert C \vert =\frac{n}{4}+y $. Therefore $\vert B \cup D \vert = \frac{n}{2}+x-y$.
By similar reasons as in case 1, the value of $f^\mathcal{R}_p$ is minimized when  $\vert B \vert =\frac{n}{2}-y$ and $\vert D \vert =x$. Therefore,
\begin{eqnarray*}
f^\mathcal{R}_p &\geq& \left(\frac{n}{4}-x\right) \left(\frac{n}{4}+y\right)+\left(\frac{n}{2}-y\right)x \\
     &\geq& \frac{n^2}{16}-2xy+\frac{n}{4} (x+y)\\
\end{eqnarray*}
The value of $f^\mathcal{R}_p$ is minimized when $\frac{n}{4}(x+y)-2xy$ is minimized. Since $\vert A \cup B \vert \geq 
\frac{n}{4}$ and $x+y \leq \frac{n}{2}$, this value is minimized when $x=y=\frac{n}{4}$. Thus, $s^\mathcal{R}(n) \geq \frac{n^2}{16}$.

For the upper bound, consider a set $P$ of $n$ points arranged uniformly along the 
boundary of a circle as in figure \ref{fig:upperBoundStrong}. Now, we claim that any point $p \in P$ is contained in at most 
$\frac{n^2}{16} $ rectangles of $\mathcal{R}$. W.l.o.g, let $p$ be a point in the top left quadrant of the circle that is $k$ 
points away from the topmost point in $P$. Let $h$ and $v$ be the horizontal and vertical lines passing through $p$. $h$ and $v$ 
divide the plane into four quadrants. Therefore $f^\mathcal{R}_p =(\frac{n}{2}-2k)2k=nk-4k^2$. This value is maximized when $k=\frac{n}{8}$. 
Thus, $s^\mathcal{R}(n) \leq \frac{n^2}{16}$.

\end{proof}

\section{Special Rectangles}
\label{spl_rec}
In this section, we prove bounds for the first selection lemma for some special families of axis-parallel rectangles.

Let $p$ be any point and  $v$ and $h$ be the vertical and horizontal lines passing through $p$, dividing the plane into four quadrants as shown in  
figure \ref{rectangle}. Let $\vert A \vert$ represent $\vert A \cap P \vert$(similar for all quadrants).

\subsection{Quadrants}
Quadrants are infinite regions defined by two mutually orthogonal halfplanes. We consider induced quadrants of a fixed orientation as shown in figure~\ref{ind_orth}. If two points are in monotonically decreasing position, then the induced quadrant is defined by two rays passing through the points (see figure~\ref{ind_orth}(a)). Otherwise, the quadrant is anchored at the point with the smaller $x$ and $y$ co-ordinate and the other point is contained in the quadrant (see figure~\ref{ind_orth}(b)). In this case, the same quadrant may be induced by different point pairs.
Let $\mathcal{O}$ represent the family of quadrants induced by a point set. Note that the family of all induced quadrants is a multiset.
\begin{figure}[h]
\begin{center}
\includegraphics[scale=0.5]{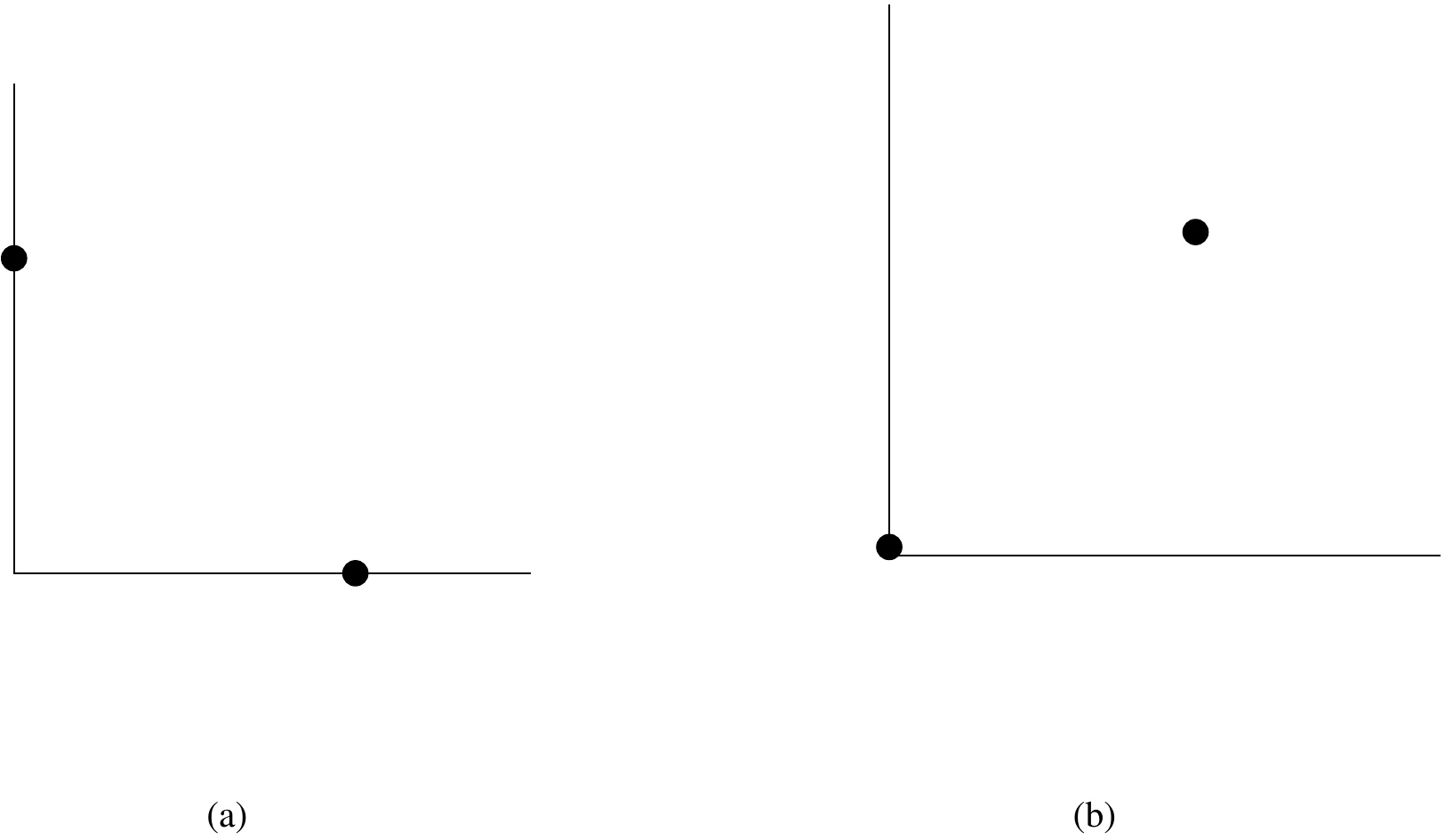}
\caption{Quadrants induced by two points}
\label{ind_orth}
\end{center}
\end{figure}

The weak variant of the first selection lemma is
trivial. Let us take the point $(x_{max}, y_{max})$, where $x_{max}$ and $y_{max}$ are the maximum values of the $x$ and $y$
coordinates of $P$. It is easy to see that this point is present in all the induced quadrants i.e. $w^\mathcal{O}(n) =
\frac{n^2}{2}$.

We also prove a tight bound for the strong variant.
\begin{lemma}\label{orth_cp}
 For any point set $P$ of $n$ points, there exists $p \in P$ such that $p$ is contained in all quadrants that contain more than $\frac{n}{2}$ points from $P$. 
\end{lemma}
\begin{proof}
Let $h$ be a horizontal line such that it has $\frac{n}{2}-1$ points of $P$ below it and $v$ be a vertical line that contains $\frac{n}{2}-1$ points of $P$ to the left of it. $h$ and $v$ divide $P$ into four quadrants as shown in figure \ref{rectangle}.  By construction, $\vert B \vert +\vert C \vert =\frac{n}{2}+1$ and $\vert C \vert \leq \frac{n}{2}-1$. Therefore $ B \cap P \ne \emptyset$. 

Let $p\in P$ be any point in $B$. Clearly any quadrant that does not contain $p$ lies completely to the right of $p$ or completely above $p$ and therefore contains at most $\frac{n}{2}$ points. Therefore, any quadrant that contains more than $\frac{n}{2}$ points from $P$ contains $p$.
%



\end{proof}

\vspace{3pt}
\begin{theorem}
$s^\mathcal{O}(n)=\frac{n^2}{4}$
\end{theorem}
\begin{proof}
Let $p\in P$ be a point as described in lemma \ref{orth_cp} i.e, $p$ is contained in all quadrants that contain more than $\frac{n}{2}$ points from $P$. We claim that $p$ is contained in at least $\frac{n^2}{4}$ induced quadrants.

Let $p$ divide the plane into four quadrants as shown in figure \ref{rectangle}.
We know that,
\begin{center}
 $\vert A \vert + \vert B \vert \leq \frac{n}{2}$
\\$\vert B \vert + \vert C \vert \leq \frac{n}{2}$
\end{center}
Assume $\vert D \vert=x$. Therefore,
$\vert A \vert,\vert C\vert \geq \frac{n}{2}-x$.
\begin{eqnarray*}
f^\mathcal{O}_p &=& \frac{{\vert D \vert}^2}{2}+\vert D \vert\left(\vert A \vert+\vert B\vert +\vert C\vert\right)+\vert A \vert .\vert C \vert \\
&\geq&\frac{x^2}{2}+ x(n-x)+\left(\frac{n}{2}-x\right)^2\\
&\geq& \frac{n^2}{4} 
\end{eqnarray*}
Therefore $p$ is contained in at least $\frac{n^2}{4}$ induced quadrants.

To prove the upper bound, consider $P$ as $n$ points arranged in a monotonically decreasing order. Let $p$ be any point in $P$. Then $p$ is contained in all quadrants induced by two points $q,r \in P$ where $q$ lies above $p$ and $r$ lies below $p$. Let $p$ be $x$ points away from the topmost point in $P$.
Therefore,
$f^\mathcal{O}_p = x(n-x)$\\
The value of $f^\mathcal{O}_p$ is maximized when $x=\frac{n}{2}$. Therefore $s^\mathcal{O}(n) \leq \frac{n^2}{4}$.
\end{proof}

\subsection{Axis-Parallel Slabs}

\begin{figure}[h]
\begin{center}
\includegraphics[scale=0.5]{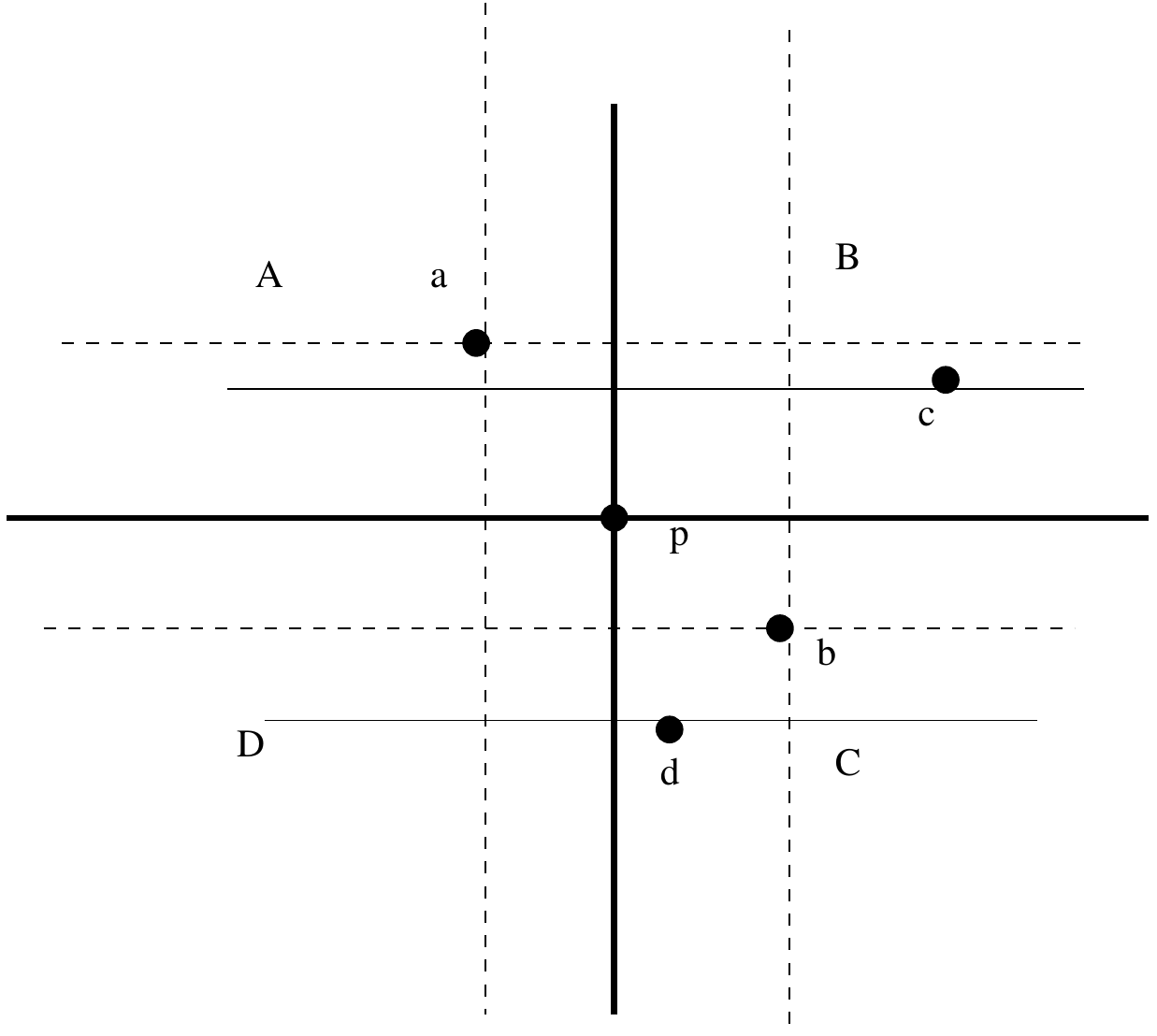}
\caption{Both slabs defined by a and b contain p, whereas only the horizontal slab defined by c and d contains p}
\label{slabfig}
\end{center}
\end{figure}

Axis-parallel slabs are a special class of axis-parallel rectangles where two horizontal or two vertical sides are unbounded. 
Each pair of points $p(x_1,y_1)$ and $q(x_2,y_2)$ induces two axis-parallel slabs of the form $[x_1,x_2]\times (-\infty,+\infty)$ 
and $(-\infty,+\infty)\times [y_1,y_2]$. Let $\mathcal{S}$ represent the family of $2 {n\choose 2}$ axis-parallel slabs induced by $P$. 

We first look at the weak variant for axis-parallel slabs. Let $x_{med}$ be the median of $P$ when the points are projected onto the $x$ axis. Similarly, let $y_{med}$ be the 
median of $P$ when the points are projected onto the $y$ axis. We claim that $(x_{med}, y_{med})$ is present in 
$\frac{n^2}{2}$ induced slabs. Indeed, $x_{med}$ is present in at least $\frac{n^2}{4}$ intervals, obtained by projecting the 
vertical slabs onto the $x$ axis. Similarly, $y_{med}$ is present in at least $\frac{n^2}{4}$ intervals, obtained by projecting 
the horizontal slabs onto the $y$ axis. Since the set of horizontal and vertical slabs are disjoint, $(x_{med}, y_{med})$ is 
present in at least $\frac{n^2}{4} + \frac{n^2}{4} = \frac{n^2}{2}$ induced slabs. It can be easily seen that this bound is 
tight.


Now we consider the strong variant. Let $p \in P$ be any arbitrary point. Note that for a pair of points $x,y \in P$, $p$ is contained in both the 
vertical and horizontal axis-parallel slabs induced by them if $x$ and $y$ are present in diagonally opposite quadrants w.r.t $p$ and in 
exactly one of the induced axis-parallel slabs if $x$ and $y$ are present in adjacent quadrants w.r.t $p$ (see figure \ref{slabfig}). Therefore,
\begin{center}
$f^\mathcal{S}_p= 2\left(\vert A \vert .\vert C\vert+\vert B\vert.\vert D \vert\right) + \left(\vert A\vert + \vert C\vert\right)\left(\vert B \vert +\vert D\vert\right)$.
\end{center}
\begin{theorem}
$ s^\mathcal{S}(n) = \frac{3n^2}{8}$
\end{theorem}
\begin{proof}

Let $p \in P$ be the strong centerpoint for axis-parallel rectangles~\cite{AGK10}. Note that this is also a strong centerpoint for 
axis-parallel slabs i.e., any axis-parallel slab that contains more than $\frac{3n}{4}$ points from $P$ contains $p$. We claim 
that $p$ is contained in at least $\frac{3n^2}{8}$ induced axis-parallel slabs. 

Let $p$ divide the plane into four quadrants as shown in figure \ref{rectangle}. If $\vert A \vert = \frac{3n}{4}$ then $\vert C \vert = \frac{n}{4}$ and $f^\mathcal{S}_p \geq \frac{3n^2}{8}$. Therefore, 
assume that $\vert A \vert = \frac{3n}{4}-x$. Assume that $x \leq \frac{n}{2}$ (There exists at least one quadrant such that this is true). Now there are two cases:
\begin{enumerate}
\item $ \vert C \vert = \frac{n}{4}-y$: 

Since $p$ is a strong centerpoint, adjacent quadrants have at least $\frac{n}{4}$ points. Therefore quadrants $B$ and $D$ should contain at least $y$ points of $P$. Also, adjacent quadrants have at most $\frac{3n}{4}$ points. Therefore quadrants $B$ and $D$ have at most $x$ points of $P$. This implies $x \geq y$.
\begin{eqnarray*}
f^\mathcal{S}_p &=&2\left(\vert B \vert\cdot \vert D \vert+\left(\frac{3n}{4}-x\right)\left(\frac{n}{4}-y\right)\right)+\left(x+y\right)\left(n-\left(x+y\right)\right)\\
\end{eqnarray*}

 $f^\mathcal{S}_p$ is minimized when $\vert B \vert. \vert D \vert$ is minimized i.e., the points are distributed as unevenly as possible between $B$ and $D$. Therefore, $f^\mathcal{S}_p$ is minimized when  $\vert B \vert =x$ and $\vert D \vert =y$. 

\begin{eqnarray*}
 f^\mathcal{S}_p &=&2\left(xy+\left(\frac{3n}{4}-x\right)\left(\frac{n}{4}-y\right)\right)+\left(x+y\right)\left(n-\left(x+y\right)\right)\\
&=&\frac{3n^2}{8}+2xy+\frac{nx}{2}-\frac{ny}{2}-x^2-y^2\\
&=&\frac{3n^2}{8}+\left (\frac{n}{2}\left(x-y\right)-\left(x-y\right)^2\right) \geq \frac{3n^2}{8}
\end{eqnarray*}
\item
$\vert C \vert = \frac{n}{4}+y$:

In this case, 
\begin{eqnarray*}
 f^\mathcal{S}_p &=&2\left(\vert B \vert\cdot \vert D \vert+\left(\frac{3n}{4}-x\right)\left(\frac{n}{4}+y\right)\right)+\left(x-y\right)\left(n-\left(x-y\right)\right)\\
\end{eqnarray*}
By reasons similar to case 1, $0\leq \vert B \vert,\vert D \vert \leq x$. The value of $f^\mathcal{S}_p$ is minimized when $ B $ or $ D $ is empty. Therefore,

\begin{eqnarray*}
 f^\mathcal{S}_p &=& 2\left(\frac{3n}{4}-x\right)\left(\frac{n}{4}+y\right)+\left(x-y\right)\left(n-\left(x-y\right)\right)\\
&=& \frac{3n^2}{8}+x\left(\frac{n}{2}-x\right)+y\left(\frac{n}{2}-y\right) \geq \frac{3n^2}{8}
\end{eqnarray*}

\end{enumerate}

To prove the upper bound, consider $P$ as $n$ points arranged along the boundary of a circle. Let $p \in P$. W.l.o.g assume 
that $p$ is $k$ points away from the topmost point and $k \leq \frac{n}{4}$. $h_p$ and $v_p$ divides the plane into four 
regions containing $2k,\frac{n}{2},\frac{n}{2}-2k,0$ points from $P$. 
Therefore,
\begin{eqnarray*}
f^\mathcal{S}_p &=& 2. 2k\left(\frac{n}{2}-2k\right) +\frac{n}{2}.\frac{n}{2}\\
&=& 2nk-8k^2+\frac{n^2}{4}
\end{eqnarray*}
The value of $f^\mathcal{S}_p$ is maximized when
\begin{center}
$2n-16k=0$\\
i.e., $k=\frac{n}{8}$\\
\end{center}
Therefore,
\begin{center}
$f^\mathcal{S}_p \leq \frac{3n^2}{8}$
\end{center}
\end{proof}

\subsection{Skylines}
Skylines are axis-parallel rectangles that are unbounded along a fixed axis, say negative $Y$ axis. A skyline induced by 
two points has the point with the higher $y$-coordinate at one corner and the other point in the opposite vertical edge 
(see figure \ref{sky_def}). Let $\mathcal{K}$ represent the family of all $n \choose 2$ skylines induced by $P$.

\begin{figure}[h]
\begin{center}
\includegraphics[scale=0.5]{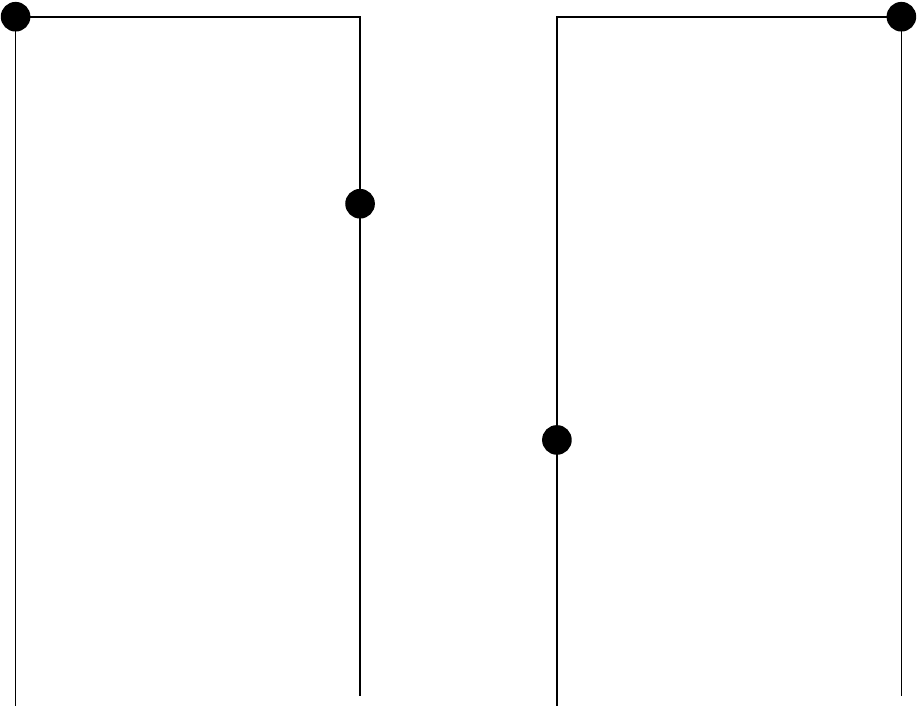}
\caption{Skyline induced by two points}
\label{sky_def}
\end{center}

\end{figure}

As in the case of induced orthants and slabs, the weak first selection lemma for skylines is straightforward. Let
$x_{med}$ be the median of $P$ projected onto the $x$ axis. Since the skylines can be assumed to be anchored on the $x$ 
axis, $x_{med}$ is present in at least $\frac{n^2}{4}$ intervals skylines. This is because $x_{med}$ is present in 
$\frac{n^2}{4}$ intervals obtained by projecting $\mathcal{K}$ on the $x$ axis. It is easy to see that this bound is tight.

For the strong variant of the first selection lemma, we prove almost tight bounds.

\begin{lemma}\label{cp_sky}
 For any set $P$ of $n$ points, there exists $p \in P$ such that any skyline that contains more than $\frac{2n}{3}$ points from $P$ contains $p$.
\end{lemma}
\begin{proof}
 Let $v_1$(resp. $v_2$) be a vertical line that has $\frac{n}{3}-1$ points of $P$ to the left(resp. right) of it. Let $h$ be a horizontal line that has $\frac{n}{3}$ points of $P$ above it. Thus we get a grid-like structure as shown in figure \ref{sky_ub}.
 

The region $E$ cannot be empty since $\vert B \vert+\vert E\vert = \frac{n}{3}+2$ and $\vert B \vert \leq \frac{n}{3}$. Let $p$ be any point in the region $E$. We claim that $p$ is contained in all skylines that contain more than $\frac{2n}{3}+1$ points from $P$. 

Any skyline $S$ that contains more than $\frac{2n}{3}+1$ points from $P$ takes points from all three vertical slabs and from both horizontal slabs. Therefore  $S$ contains the entire region $ E $ and therefore the point $p$.

%
%
%

\begin{figure}[h]

 \begin{center}
 \scalebox{0.5}{ 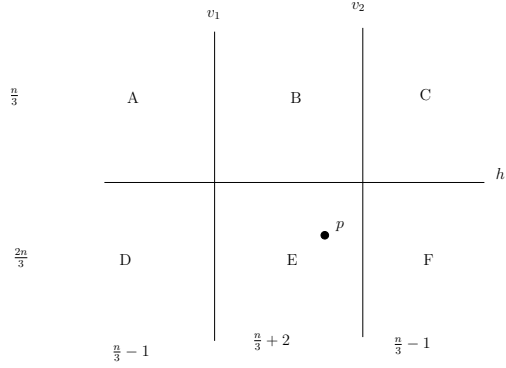}
 \end{center}
\caption{To prove lemma \ref{cp_sky}}
\label{sky_ub}

\end{figure}
\end{proof}

\begin{theorem}
$\frac{n^2}{9} \leq s^\mathcal{K}(n) \leq \frac{n^2}{8}$
\end{theorem}
Let $p \in P$ be a point as described in lemma \ref{cp_sky} i.e., any skyline that contains more than $\frac{2n}{3}$ points 
from $P$ contains $p$. We claim that $p$ is contained in at least $\frac{n^2}{9}$ induced skylines. 

Let $p$ divide the plane into four quadrants as shown in figure \ref{rectangle}. Therefore,

\begin{center}
$f_p^{\mathcal{K}}=\vert A \vert \vert C \vert + \vert B \vert \vert D \vert + \vert A \vert \vert B \vert$
\end{center}

If both $\vert A \vert$ and $\vert C \vert$ are $\geq \frac{n}{3}$ then the claim is true. Therefore assume this is not true. 
Now there are four cases. In all the cases, we fix the number of points in $A$ and $C$. Note that the value of $f^\mathcal{K}_p$ 
is minimized when $B$ has very few points than $D$.
\begin{enumerate}
\item $\vert A \vert=\frac{n}{3}-x, \vert C \vert=\frac{n}{3}-y, x \leq y$ : 
\begin{center}
$f^\mathcal{K}_p=\left(\frac{n}{3}-x\right)\left(\frac{n}{3}-y\right)+ \vert B\vert \vert D\vert+\left(\frac{n}{3}-x\right)\vert B\vert$
\end{center}

Since $\frac{n}{3} \leq \vert A \vert + \vert D \vert \leq \frac{2n}{3}$ and $\vert C \vert + \vert D \vert \leq \frac{2n}{3}$,
\begin{equation*}
y \leq \vert D \vert \leq \frac{n}{3}+x
\end{equation*}
Also,$ \vert A \vert+\vert B \vert \geq \frac{n}{3}$ and $\frac{n}{3} \leq \vert B \vert + \vert C \vert \leq \frac{2n}{3}$. 
Therefore,
\begin{equation*}
y \leq \vert B \vert \leq \frac{n}{3}+y
\end{equation*}
The value of $f^\mathcal{K}_p$ is minimized when $\vert B \vert =y$ and $\vert D \vert =\frac{n}{3}+x$. Therefore,

\begin{eqnarray*}
f^\mathcal{K}_p&=&\left(\frac{n}{3}-x\right)\left(\frac{n}{3}-y\right)+ y\left(\frac{n}{3}+x\right)+\left(\frac{n}{3}-x\right)y\\
&=&\frac{n^2}{9}+xy + \frac{n}{3}\left(y-x\right)\\
& \geq & \frac{n^2}{9}
\end{eqnarray*}
\item $\vert A \vert=\frac{n}{3}-x, \vert C \vert=\frac{n}{3}-y, x \geq y$ : 

Since $\frac{n}{3} \leq \vert A \vert + \vert D \vert \leq \frac{2n}{3}$ and $\vert C \vert + \vert D \vert \leq \frac{2n}{3}$,
\begin{equation*}
x \leq \vert D \vert \leq \frac{n}{3}+y
\end{equation*}
Also,$ \vert A \vert+\vert B \vert \geq \frac{n}{3}$ and $\frac{n}{3} \leq \vert B \vert + \vert C \vert \leq \frac{2n}{3}$. 
Therefore,
\begin{equation*}
x \leq \vert B \vert \leq \frac{n}{3}+y
\end{equation*}
The value of $f^\mathcal{K}_p$ is minimized when $\vert B \vert =x$ and $\vert D \vert =\frac{n}{3}+y$.
\begin{eqnarray*}
f^\mathcal{K}_p&=&\left(\frac{n}{3}-x\right)\left(\frac{n}{3}-y\right)+x\left(\frac{n}{3}+y\right)+x\left(\frac{n}{3}-x\right)\\
&=&\frac{n^2}{9}+\left(x-y\right)\left(\frac{n}{3}+y-x\right)+y^2\\
& \geq & \frac{n^2}{9}
\end{eqnarray*}

\item  $\vert A \vert=\frac{n}{3}-x, \vert C \vert=\frac{n}{3}+y$ : 

By reasons similar to case 2,

\begin{eqnarray*}
x \leq \vert B \vert \leq \frac{n}{3}-y\\
x \leq \vert D \vert \leq \frac{n}{3}-y\\
\end{eqnarray*}
Therefore, the value of $f^\mathcal{K}_p$ is minimized when $\vert B \vert=x$ and $\vert D \vert=\frac{n}{3}-y$. Since 
$x \leq \frac{n}{3}$, this case now becomes exactly like one of the previous cases where two diagonally opposite quadrants have 
less than $\frac{n}{3}$ points.

\item  $\vert A \vert=\frac{n}{3}+x, \vert C \vert=\frac{n}{3}-y$ : 

Here
$\vert B \vert +\vert D \vert=\frac{n}{3}+y-x$. Also,
\begin{eqnarray*}
y \leq \vert B \vert \leq \frac{n}{3}-y\\
\vert D \vert \leq \frac{n}{3}-x
\end{eqnarray*}
Therefore, the value of $f^\mathcal{K}_p$ is minimized when $\vert B \vert =y$ and $\vert D \vert =\frac{n}{3}-x$. Since 
$y \leq \frac{n}{3}$, this becomes exactly like case 1 or 2.
\end{enumerate}
Therefore,
\begin{center}
 $ s^\mathcal{K}(n) \geq \frac{n^2}{9}$
\end{center}

\begin{figure}[t]

\begin{centering}
\scalebox{0.5}{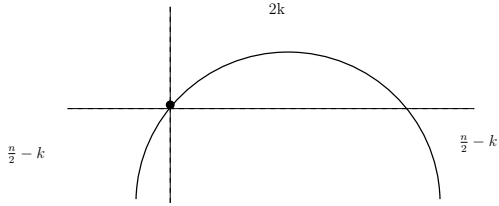}
\caption {Upper bound construction for skylines}
\label{sky}
\end{centering}
\end{figure}

To show the upper bound we consider $P$ as shown in figure \ref{sky}. $n$ points are arranged along the boundary of a 
semicircle. Let $p$ be any point in $P$. We claim that $p$ is contained in at most $\frac{n^2}{8}$ induced skylines. 

Assume that $p$ is the $k$th point from the topmost point. Therefore,
\begin{eqnarray*}
f^\mathcal{K}_p &=& 2k\left(\frac{n}{2}-k\right)\\
&=& nk-2k^2
\end{eqnarray*}
The value of $f^\mathcal{K}_p$ is maximized when $k=\frac{n}{4}$ and $f^\mathcal{K}_p \leq \frac{n^2}{8}$.

\section{Boxes in $\mathbb{R}^d$}
\label{boxes}
\newcommand{\rpm}{\raisebox{.2ex}{$\scriptstyle\pm$}}

Let $P$ be a set of $n$ points in $\reals^d$ and $\mathcal{B}$ be the set of all $n\choose2$ boxes induced by $P$. Let $B(a,b)$ 
be the box induced by $a,b \in P$ i.e, box $B(a,b)$ has $a$ and $b$ as diagonal points. We define $\mathcal{B}_p \subseteq \mathcal{B}$ as the set of boxes which contain 
a point $p \in \reals^d$. We look at a lower bound for $w^{\mathcal{B}}(n)$.

\begin{theorem}
\label{SelectionBoxesRdLower}
For $d \geq 2$, $w^{\mathcal{B}}(n) \geq \frac{n^2}{2^{(2^d - 1)}}$.
\end{theorem}

\begin{proof}
We prove this by induction on the dimension $d$. The base case $d=2$ is true from Theorem 
\ref{thm:selectionRectanglesLower}.

For $d \geq 3$, we assume that the statement is true for induced boxes in $d-1$ dimensions. 
We project all the points of $P$ orthogonally onto $h$, which is a $(d-1)$ dimensional hyperplane $x_d=0$. From the 
induction hypothesis, there exists a point $q = (q_1, \dots, q_{d-1})$ in this hyperplane which is present in 
$\frac{n^2}{2^{2^{(d-1)} - 1}}$ of the $(d-1)$ dimensional boxes induced by the projections of $P$ on $h$.

Consider the line perpendicular to the hyperplane $h$, which passes through $q$. The line $l$ passes through those $d$
dimensional boxes, whose projections onto $h$ contained $q$. We project only these boxes onto $l$ and look at
the problem of second selection lemma for intervals ($d=1$), where the number of points is $n$ and the number of 
intervals is $\frac{n^2}{2^{2^{(d-1)} - 1}}$. From lemma \ref{lem:secondSelectionIntervals}, we see that there exists 
a point $q_d$ in $\frac{n^2}{2^{2^d - 1}}$ intervals. This in turn gives us a point $r = (q_1, \dots, q_d)$ which is 
present in the corresponding boxes in $\reals^d$.
\end{proof}

\begin{theorem}
\label{SelectionBoxesRdUpper}
For $d \geq 2$, $w^{\mathcal{B}}(n) \leq \frac{n^2}{2^{d+1}} + o(n^2)$
\end{theorem}

\begin{proof}
Consider a set of $n$ points which is arranged as a uniform $2^k \times 2^k \times ....\times 2^k$ $d$-dimensional 
grid, where $k =\frac{1}{d}$. \footnote{We eliminate all the degenerate rectangles i.e. the set of all rectangles induced 
by two points which are along any row of the grid in each dimension. Note that for $d \geq 2$, the number of such 
degenerate rectangles is at most $dn^{(1+\frac{1}{d})} = o(n^2)$.}

Construct $d$ hyperplanes parallel to the coordinate axes, $H = \{h_1, h_2, ..., h_d\}$ which are perpendicular to 
each other, each of which divides the grid into 2 halves, containing $\frac{n}{2}$ points each. Let the intersection 
point of these $d$ hyperplanes be $p$. Now, each of the $2^d$ orthants defined by $H$ contains a smaller $d$-dimensional 
uniform grid of size $\frac{n^{\frac{1}{d}}}{2} \times \frac{n^{\frac{1}{d}}}{2} \cdot \cdot \cdot \times 
\frac{n^{\frac{1}{d}}}{2}$. Thus, each of the orthants contain exactly $\frac{n}{2^d}$ points. The number of boxes which 
contain $p$ is given by $(2^{d-1} \cdot \frac{n^2}{2^{2d}})$. The term $2^{d-1}$ is the number of opposite orthant pairs, 
whose points contribute to a box containing $p$. Thus, $|\mathcal{B}_p|=\frac{n^2}{2^{d+1}}$.

Consider any point $q \in \reals^d$ (not necessarily from $P$), which is present inside the grid. Construct $d$ orthogonal 
hyperplanes $L = \{l_1, l_2, ...., l_d\}$ parallel to $H$, all of which intersect at $q$. Let $r_k$ be the number of grid 
points in the $k^{th}$ dimension between $p$ and $q$ ($h_i$ and $l_i$).

Consider the $d$-dimensional uniform grids present in each of the orthants formed by $L$. Let us fix a dimension $k$,
where $k \in [d]$. Consider any orthant $O$ realized by $L$ and let $G'$ be the grid present in $O$. Let $n_1$ be the 
number of points present in the $k^{th}$ dimension in $O$. This means that the diagonally opposite orthant $O'$ to 
$O$ contains $n_2 = n^{\frac{1}{d}} - n_1$ points in the $k^{th}$ dimension. W.l.o.g, let 
$n_1 = \frac{n^{\frac{1}{d}}}{2}-r_k$ and thus, $n_2=\frac{n^{\frac{1}{d}}}{2}+r_k$, where $0 \leq r_i \leq 
\frac{n^{\frac{1}{d}}}{2}, \forall i \in [d]$. This is true for points along any dimension $1 \leq k \leq d$, in any 
orthant defined by $L$. W.l.o.g, let $G'$ be of size $\displaystyle \prod_{i=1}^d \left(\frac{n^{\frac{1}{d}}}{2}-r_i\right)$. $O'$ 
will then have a grid of size $\displaystyle \prod_{i=1}^d \left(\frac{n^{\frac{1}{d}}}{2}+r_i\right)$. Thus, the number of induced 
boxes contributed to $\mathcal{B}_q$ by diagonally opposite orthants $O$ and $O'$, is $\displaystyle \prod_{i=1}^d 
(\frac{n^{\frac{1}{d}}}{2}-r_i)(\frac{n^{\frac{1}{d}}}{2}+r_i) = \prod_{i=1}^d (\frac{n^{\frac{2}{d}}}{4}-r_i^2)$. Since, 
this is true for every octant (having different combinations of$(\frac{n^{\frac{1}{d}}}{2} \rpm r_i), \forall i \in [d]$), 
we get the same term in $\mathcal{B}_q$ for every pair of opposite orthants. The number of such orthant pairs is $2^{d-1}$ 
and thus, $\vert \mathcal{B}_q \vert$ is given by -

\begin{equation*}
\begin{split}
\vert \mathcal{B}_q \vert & = 2^{d-1} \cdot (\frac{n^{\frac{2}{d}}}{4} - r_1^2) \cdot (\frac{n^{\frac{2}{d}}}{4} - r_2^2) 
\dots (\frac{n^{\frac{2}{d}}}{4} - r_d^2) \\
\implies \vert \mathcal{B}_q \vert & \leq  \frac{n^2}{2^{d+1}}
\end{split}
\end{equation*}

The point $q$ is chosen arbitrarily and thus, any point in $\reals^2$ is present in at most $\frac{n^2}{2^{(d+1)}}$ induced
boxes.
\end{proof}

\section{Hyperspheres in $\mathbb{R}^d$}
\label{disks}

Let $P$ be a set of $n$ points in $\reals^d$ and $\mathcal{C}$ be the set of $n \choose 2$ hyperspheres induced by $P$. Let $C(a,b)$ 
be the hypersphere induced by $a,b \in P$ i.e, $C(a,b)$ has $a$ and $b$ as diametrically opposite points.  

\subsection{Weak Variant for hyperspheres in $\reals^d$}

In this section, we obtain bounds for $w^\mathcal{C}(n)$.
\subsubsection{General Point Sets}
\begin{lemma}\label{circ_lemm}
$w^\mathcal{C}(n) \geq \frac{n^2}{2(d+1)}$
\end{lemma}
\begin{proof}
Let $c$ be the centerpoint of $P$. Therefore any halfspace that contains $c$ contains at least $\frac{n}{d+1}$ points. 
We claim that $c$ is contained in at least $\frac{n^2}{2(d+1)}$ induced hyperspheres.

Let $p$ be any point in $P$. Let $H$ be the halfspace that contains $c$ and whose outward normal is $\vec{cp}$. $H$ 
contains at least $\frac{n}{d+1}$ points from $P$. Now, $c$ is contained in a hypersphere induced by $p$ and any point 
$p_1$ in $H$ since $\angle pcp_1 > 90\degree$. Thus $c$ is contained in at least $\frac{n}{d+1}$ induced hyperspheres 
where one of the inducing points is $p$. As this is true for any point in $P$, $c$ is contained in $\frac{n^2}{2(d+1)}$ 
induced hyperspheres.  
\end{proof}

The upper bound construction is a trivial one and comes from the arrangement of $P$ as a monotonically increasing line in 
$\reals^d$. This gives us that any point $p \in \reals^d$ is present in at most $\frac{n^2}{4}$ hyperspheres.
\subsubsection{Centrally Symmetric Point Set}
In this section, we prove tight bounds for a special class of point sets viz. centrally symmetric point sets. Let $P$ be a centrally symmetric point set w.r.t origin i.e., for any point $p \in P$, $-p$ also belongs to $P$.

\begin{theorem}
 $w^\mathcal{C}(n) =\frac{n^2}{4}$
\end{theorem}
\begin{proof}
 The proof is similar to that of lemma \ref{circ_lemm}. 
 
 Let $o$ be the origin of the centrally symmetric point set $P$. Let $p$ be any point in $P$. Let $H$ be the halfspace that contains $o$ and whose outward normal is $\vec{op}$. $H$ 
contains $\frac{n}{2}$ points from $P$ since for any point $p_1 \in P\setminus (H \cap P)$, $-p \in H \cap P$. By reasons similar to lemma  \ref{circ_lemm}, $o$ is contained in  at least $\frac{n^2}{4}$ induced hyperspheres. 

To prove the upper bound, consider points arranged uniformly along a monotonically increasing line in $\mathbb{R}^d$.
\end{proof}

\subsection{Strong Variant for disks in $\mathbb{R}^2$}
In this section, we obtain bounds on $s^\mathcal{C}(n)$ when $\mathcal{C}$ is the family of induced disks in $\mathbb{R}^2$. 
\subsubsection{General Point Sets}
\begin{theorem}
 $\frac{n^2}{16} \leq s^\mathcal{C}(n) \leq \frac{n^2}{9}$
\end{theorem}
\begin{proof}
 The lower bound follows from theorem \ref{strong_rec} since the axis-parallel rectangle induced by two points $p,q$ are completely contained inside the disk induced by $p$ and $q$.

To prove the upper bound, we use a configuration from \cite{CCEG79}. $n$ points are arranged as equal subsets of $\frac{n}{3}$ points, each along small circular arcs at the vertices of a triangle $\triangle ABC$. Let $A_1=\{a_1,a_2,\cdots,a_{\frac{n}{3}}\}$ represent the points near the vertex $A$. Similarly, let $B_1=\{b_1,b_2,\cdots,b_\frac{n}{3}\}$ represent the points near vertex $B$ and $C_1=\{c_1,c_2,\cdots,c_\frac{n}{3}\}$ represent the points near vertex $C$. The angles of the triangle and the length of the arcs are so selected such that the only obtuse-angled triangles are of type $\triangle a_ia_ja_k,\triangle b_ib_jb_k,\triangle c_ic_jc_k,\triangle a_ib_jb_k,\triangle b_ic_jc_k,\triangle c_ia_ja_k$ where $1 \leq i,j,k \leq \frac{n}{3}$(See section 5 in \cite{CCEG79}).

We claim that any point $p \in P$ is contained in at most $\frac{n^2}{9}$ induced disks. W.l.o.g assume that $p\in A_1$. Also assume that $p$ has $x$ points of $A_1$ above it(i.e, away from $C$). The triangle with one vertex as $p$ is obtuse when both the other two vertices are from $A_1$ or $B_1$ or one of them is from $A_1$ and the other is from $C_1$. When both the vertices are from $B_1$, the angle subtended at $p$ is acute. The angles are obtuse in the following cases:
\begin{enumerate}
 \item The other two vertices are $a_i$ and $a_j$, $1\leq i,j \leq \frac{n}{3}$ and $a_i$ and $a_j$ lies on either side of $p$ in $A_1$.
\item The other two vertices are $a_i$ and $c_j$, $1\leq i,j \leq \frac{n}{3}$ and $a_i$ lies above $p$ in $A_1$.
\end{enumerate}
Therefore,
\begin{eqnarray*}
 f^\mathcal{C}_p&=& x(\frac{n}{3}-x)+\frac{n}{3}(\frac{n}{3}-x)\\
&=& \frac{n^2}{9}-x^2
\end{eqnarray*}

The value of $f^\mathcal{C}_p$ is maximized when $x=0$. Therefore,
\\$f^\mathcal{C}_p \leq \frac{n^2}{9}$.

\end{proof}
\subsubsection{Centrally Symmetric Point sets}
In this section, we prove tight bounds for centrally symmetric point sets. Let $P$ be a centrally symmetric point set w.r.t origin. 

\begin{theorem}
 $s^\mathcal{C}(n) =\frac{n^2}{8}$
\end{theorem}
\begin{proof}

\textbf{Lower Bound}
\\

 Let $P$ be a centrally symmetric point set. We claim that there exists a point $p \in P$ such that $p$ is contained in $\frac{n^2}{8}$ disks induced by $P$.

We find the point $p \in P$  as follows. Let $P_1=P$. For $i \in [1, \frac{n}{2}]$, let $a_i \in P_i$ be the point with maximum distance from the origin and let $b_i =-a_i$. The disk induced by $a_i$ and $b_i$ contains all the points of $P_i$. Otherwise, if there is a point $a_j\in P_i$ outside this disk then the distance from $a_j$ to origin is more than the distance from $a_i$ to origin, a contradiction.  Let $P_{i+1} = P_i\setminus \{a_i,b_i\}$. Since $b_i=-a_i$, $P_{i+1}$ is also centrally symmetric. Let $p \in P_{ n/2} $. Then $p$ has the desired property.

 Let $q \in P_{j+1}$ . Then we claim that $q$ is contained in at least $\frac{j^2}{2}$ induced disks. 

Let $i < j$. Clearly $q$ is contained in $C_{a_ib_i}$ and $C_{a_jb_j}$. We claim that $q$ is also contained in $C_{ab}$ where $a,b \in \{a_i,a_j,b_i,b_j\}$ and $C_{ab}$ is not $C_{a_ib_i}$ or $C_{a_jb_j}$. Assume for contradiction that this if false. Therefore $\angle a_iqa_j,\angle a_iqb_j, \angle b_iqa_j,\angle b_iqb_j$ are all acute. Consider the line segment joining $a_i$ and $q$. Let $h_a$ be the line  perpendicular to this line segment and passing through $q$. Let $H_a$ be the halfspace defined by $h_a$ containing the point $a_i$(See figure \ref{centrally}). Since angles $\angle a_iqa_j$ and $\angle a_iqb_j$ are acute, both $a_j$ and $b_j$ belong to $H_a$. Now consider the line segment joining $b_i$ and $q$. Define $H_b$ as before. By similar reasoning as before, $a_j$ and $b_j$ belong to $H_b$. Therefore, both $a_j$ and $b_j$ belong to $H_a \cap H_b$. This contradicts the fact that $\angle a_jqb_j$ is obtuse. Therefore, at least one of the angles $\angle a_iqa_j,\angle a_iqb_j, \angle b_iqa_j,\angle b_iqb_j$ is obtuse and the disk induced by the corresponding points contains $q$.


\begin{centering}
 \begin{figure}
\begin{minipage}{0.55\linewidth}

 \scalebox{0.5}{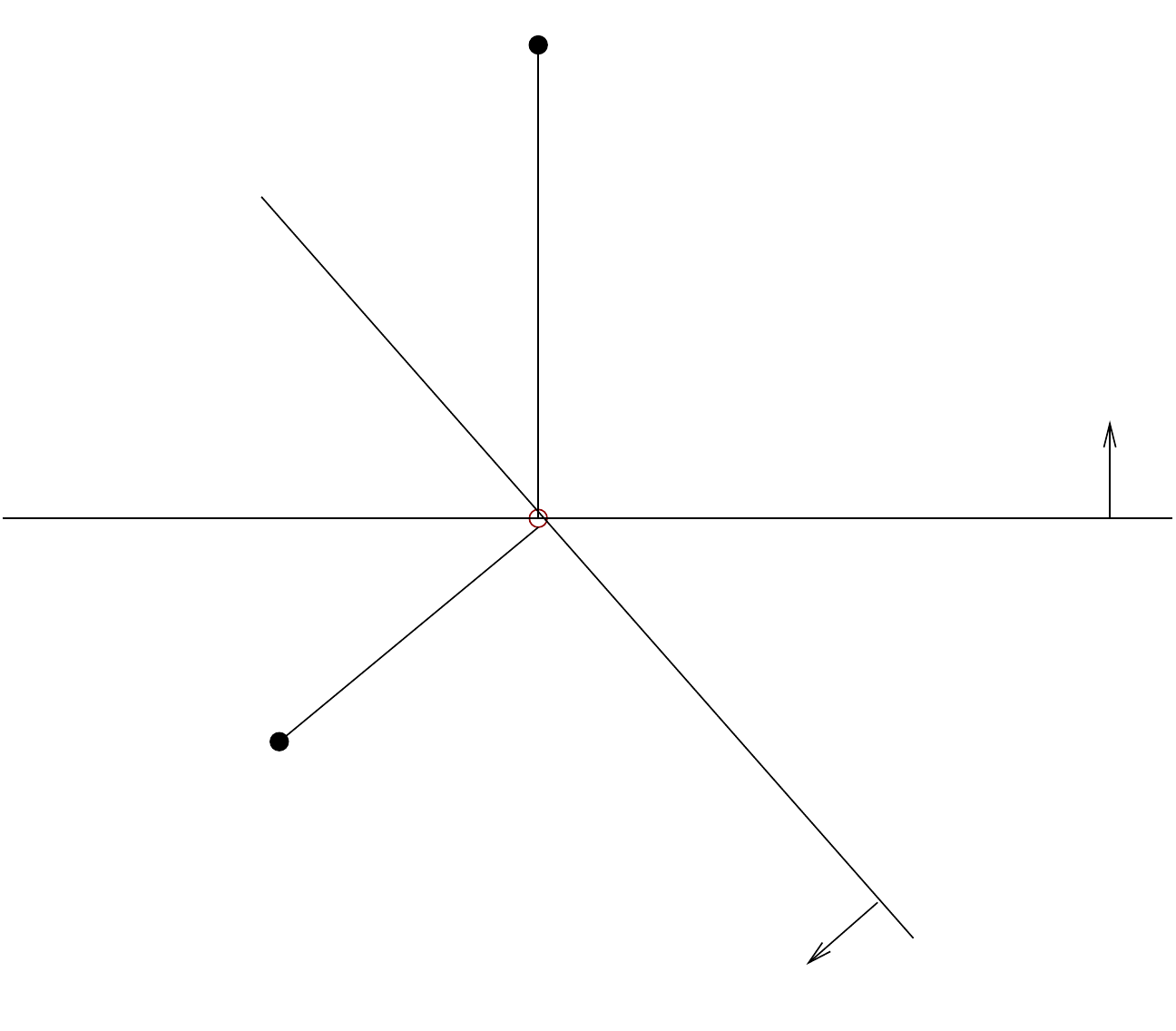}
 
 \caption{Lower Bound for disks}
 \label{centrally}
\end{minipage}
\begin{minipage}{0.45\linewidth}

 \includegraphics[scale=0.5]{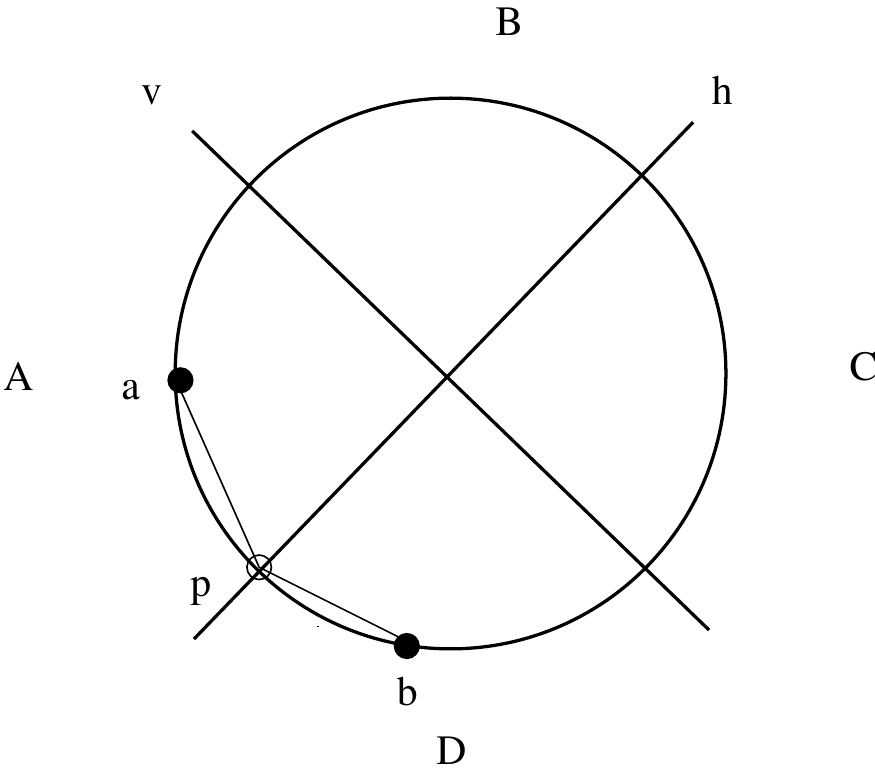}
\caption{Upper bound for Circles}
\label{cir_ub}

\end{minipage}
\end{figure}
\end{centering}
Therefore, $q$ is contained in all disks of the form $C_{a_ib_i}$ where $1\leq i \leq j$. Also, we just proved that for any $i,k \leq j$, $q$ is contained in at least one disk of the form $C_{ab}$ where $a\in \{a_k,b_k\}$ and $b\in\{a_i,b_i\}$. Therefore $q$ is contained in $\frac{j^2}{2}$ induced disks.

Since $p \in P_{n/2}$, $p$ is contained in $\frac{n^2}{8}$ induced disks.
\vspace{10pt}
\\

\textbf{Upper Bound}
\vspace{5pt}
\\
Consider $P$ as $n$ points arranged along the boundary of a circle. $P$ is a centrally symmetric point set. We claim that any point $p \in P$ is contained in at most $\frac{n^2}{8}$ induced disks.

Let $h$ be a straight line connecting $p$ and its diametrically opposite point and $v$ be a straight line perpendicular to $h$. Let $h$ and $v$ divide the plane into four quadrants as shown in figure \ref{cir_ub}. Let $a,b \in P$. If both $a$ and $b$ lie in the same side of $h$, $ \angle apb < 90$ and $p$ is not contained in the disk induced by $a$ and $b$. Therefore, assume that $a$ and $b$ lie on different sides of $h$. W.l.o.g assume that $a\in A\cup B$ and $b \in C \cup D$. Let $a$ be the $j$th point from $p$(clockwise) and $b$ be the $k$th from $p$(anti-clockwise), $j,k \in [1,\frac{n}{2}-1]$. It can be clearly seen that $\angle apb \geq 90$ when $j\in [1,\frac{n}{2}-k]$. Therefore, $p$ is contained in $1+2+...+\frac{n}{2}=\frac{n^2}{8}$ induced disks.

\end{proof}

\section{Second Selection lemma}
\label{ssl}
In the second selection lemma, we are given an arbitrary subset $\mathcal{S} \subseteq \mathcal{R}$ of size $m$. We bound 
the maximum number of induced rectangles of $\mathcal{S}$ that can be pierced by a single point $p$. The main idea of our 
approach is an elegant double counting argument, which we first illustrate for the special case of intervals in $\reals$.

\subsection{Second selection lemma for intervals in $\reals$}
Let $P = \{x_1, x_2, ..., x_n\}$ be a set of $n$ points in $\reals$. For any two points $p < q$ on the real line, we call
$[p, q]$ as the interval defined by the points $p$ and $q$. Let $C$ be the given set of $m$ intervals which are induced by
$P$, where $m \leq \binom{n}{2}$.

Let $J_c$ denote the number of points from $P$ present in an interval $c \in I$ and $I_p$ denote the number of intervals in $C$
containing the point $p$. Let us partition $C$ in such a way that, each interval with the point $x_i$ as its left endpoint is
placed in a set of intervals $X_i, \forall x_i \in P$. The intervals in $X_i$ are ordered by their right endpoint. Let each
$|X_i|$ be $m_i$ and hence $\sum_{i=1}^{n} m_i = m$.

\begin{lemma}
Let $P = \{x_1, ..., x_n\}$ be a set of $n$ points in $\reals$ and $C$ be a set of $m$ intervals induced from $P$. If 
$m=\Omega(n)$, then there exists a point $p \in P$ which is present in at least $\frac{m^2}{2n^2}+\frac{3m}{2n}$ intervals 
of $C$.
\label{lem:secondSelectionIntervals}
\end{lemma}

\begin{proof}
First, let us find the lower bound for the number of points present in all the intervals in $X_i$. In $X_i$, we can see that
the $j^{th}$ interval contains at least $j+1$ points. Thus, the summation of the number of points present in the intervals of $X_i$ 
is given by
\begin{equation*}
\begin{split}
\displaystyle\sum_{r \in X_i} J_r & \geq 2 + 3 + ... + (m_i+1) \geq \frac{m_i^2}{2} + \frac{3m_i}{2}
\end{split}
\end{equation*}

Each interval belongs to a unique $X_i$ and thus, the summation of the number of points present in the intervals in $C$ is lower 
bounded by summing over all $x_i$, the number of points present in each $X_i$.
\begin{equation*}
\begin{split}
\displaystyle\sum_{c \in I} J_c & \geq \displaystyle\sum_{i=1}^n (\frac{m_i^2}{2} + \frac{3m_i}{2}) \\
 & \geq \frac{\displaystyle\sum_{i=1}^n m_i^2}{2} + \frac{3 \cdot \displaystyle\sum_{i=1}^n m_i}{2}
\end{split}
\end{equation*}
Now, from the Cauchy-Schwarz inequality in $\reals^n$ we have, $(\sum_{j=1}^n m_j^2) \geq \frac{m^2}{n}$. \\
Thus, $\displaystyle\sum_{c \in I} J_c  \geq \frac{m^2}{2n} + \frac{3m}{2}$ \\

Now, the count we are achieving by summing over the number of points present in an interval $J_c$, can also be gotten
through by summing over the number of intervals containing a point $I_p$.
\begin{equation*}
\begin{split}
\displaystyle\sum_{c \in I} J_c & = \displaystyle\sum_{p \in P} I_p \\
\implies \displaystyle\sum_{p \in P} I_p & \geq \frac{m^2}{2n} + \frac{3m}{2}\\
\end{split}
\end{equation*}
By the pigeonhole principle, there exists a point $p \in P$ present in at least $\frac{m^2}{2n^2} + \frac{3m}{2n}$ intervals.
\end{proof}

\begin{lemma}
There exists a point set $P$ of size $n$ and a set of induced intervals $C$ of size $m \leq n^2(\sqrt{2}-1) - \frac{n}{\sqrt{2}}$, 
such that any point in $P$ is present in at most $\frac{m^2}{n^2}+\frac{3m}{\sqrt{2}n}$ intervals in $C$.
\label{lem:upperIntervals}
\end{lemma}

\begin{proof}
Let $P = \{x_1, x_2, \dots, x_n\}$ where $x_1 < x_2 < \dots < x_n$. Let $m$ be a multiple of $n$ and let $m_i = 
\frac{\sqrt{2}m}{n}$. Let the induced intervals from $C$ be of the form $[x_i, x_{i+1}], [x_i, x_{i+2}], ..., [x_i, x_{i+k}], 
\forall x_i \in P$, where $k = min(\frac{\sqrt{2}m}{n}, n-i)$. We now have, 
$|C| = (n - \frac{\sqrt{2}m}{n}) \cdot \frac{\sqrt{2}m}{n} + ((\frac{\sqrt{2}m}{n}-1) + \cdot \cdot \cdot + 1) = 
\sqrt{2}m - (\frac{m^2}{n^2} + \frac{m}{\sqrt{2}n})$. It is easy to see that $|C| \geq m$, when $m \leq n^2(\sqrt{2}-1) 
- \frac{n}{\sqrt{2}}$.

Let $B \subset P$ be the set of points, which exclude the first and the last $\frac{\sqrt{2}m}{n}$ points from $P$. Consider 
any point $x_p \in B$. Let us count the number of intervals containing $x_p$ i.e $I_{x_p}$.
\begin{figure}[h]
\centering
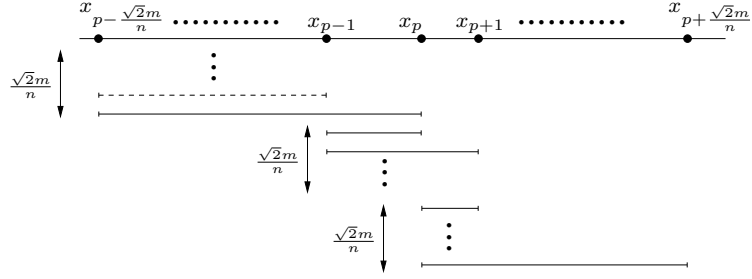
\caption {The intervals which contain $x_p$ (Bold intervals)}
\label {fig:intUpperBound}
\end{figure}

From the figure \ref{fig:intUpperBound}, it can be seen that there is only one interval from 
$I_{x_{(p-\frac{\sqrt{2}m}{n})}}$ which contains $x_p$. This count of intervals containing $x_p$ increases by 1 for each 
consecutive point after $x_{p-\frac{\sqrt{2}m}{n}}$, until we reach $x_{p-1}$ and $x_p$, both of which have 
$\frac{\sqrt{2}m}{n}$ intervals containing $x_p$. \\Thus, we have
\begin{equation*}
\begin{split}
I_{x_p} & = 1 + 2 + ... + \frac{\sqrt{2}m}{n} + \frac{\sqrt{2}m}{n} \\
        & = \frac{m^2}{n^2} + \frac{3m}{\sqrt{2}n}
\end{split}
\end{equation*}
From our construction of $C$, it can be seen that any point $q \in P-B$ will be involved in lesser number of intervals and
thus, $|I_q| < \frac{m^2}{n^2}$. The bounds are tight upto a multiplicative constant.
\end{proof}

\vspace*{-0.1cm}
\subsection{Second selection lemma for Axis-Parallel Rectangles in $\reals^2$}
Let $P$ be a set of $n$ points in $\reals^2$. Let $\mathcal{S} \subseteq \mathcal{R}$ be any set of $m$ induced axis-parallel 
rectangles. In the second selection lemma, we bound the maximum number of induced rectangles of $\mathcal{S}$ that can be 
pierced by a single point $p$. The main idea of our approach is an elegant double counting argument.

Let $R(p,q)$ denote the rectangle induced by the points $p$ and $q$. $\mathcal{S}$ is partitioned into sets $X_i$ 
as follows : any rectangle $R(x_i,u) \in \mathcal{S}$ where $x_i, u \in P$, is added to the partition $X_i$ if $u$ is higher 
than $x_i$. Let $P_i = \{u  \vert R(x_i, u) \in X_i\}$. Let $|P_i|=|X_i|=m_i$. Any rectangle $R(x_i,u) \in X_i$ is placed in 
one of two sub-partitions, $X_i'$ or $X_i''$, depending on whether $u$ is to the right or left of $x_i$. Let $|X_i'| 
= m_i'$ and $|X_i''| = m_i''$. Similarly, we partition $P_i$ into $P_i'$ and $P_i''$. Let $\sum_{i=1}^n m_i' = m'$ and 
$\sum_{i=1}^n m_i'' = m''$. The rectangles in $X_i'$ (or $X_i''$) and the points in $P_i'$ (or $P_i''$) are ordered by 
decreasing y-coordinate.

We construct a grid out of $P$ by drawing horizontal and vertical lines through each point in $P$. Let the resulting set 
of grid points be $G$ ($P \subset G$), where $|G|=n^2$. We use the grid points in $G$ as the candidate set of points for 
the second selection lemma. 

Let $J_r$ be the number of grid points in $G$ present in any rectangle $r \in \mathcal{S}$. W.l.o.g consider the set of 
rectangles present in $X_i'$. We obtain a lower bound on $\sum_{r \in X_i'} J_r$. 

\begin{lemma}
\label{lemma:rectLower}
$\displaystyle \sum_{r \in X_i'} J_r \geq \frac{(m_i')^3}{6}$.
\end{lemma}

\begin{proof}
Let $c = \sum_{r \in X_i'} J_r$. We prove the lemma by induction on the size of $m_i'$. For the base case, let $m_i' = 2$. There 
are only two ways in which the point set can be arranged, as shown in figure \ref{fig:rectLower}(a). It can be seen that the 
statement is true for the base case.

\begin{figure}[h]
\centering
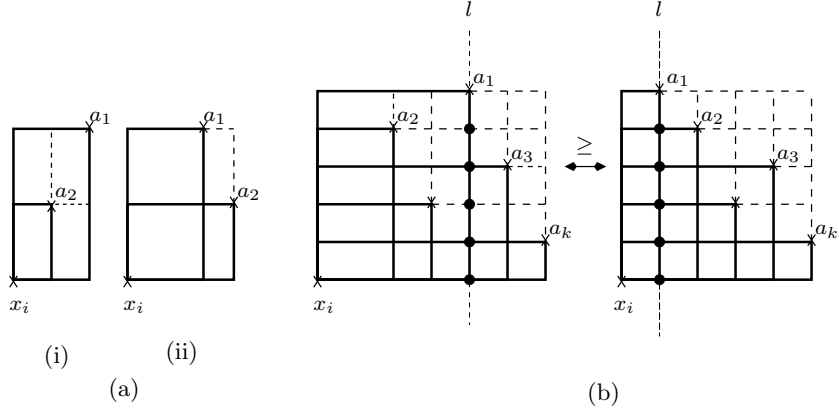
\caption {The dotted lines represent the grid lines and the solid lines represent the rectangle edges. (a) Base cases. 
(b) Inductive case - the case when $a_1$ is not the leftmost point in $P_i'$.}
\label {fig:rectLower}
\end{figure}

For the inductive case, assume that the statement is true for $m_i' = k-1$ and let $m_i'= k$. Let 
$P_i' = \{a_1, a_2, ..., a_k\}$. Let $a_1$ be the topmost point in $P_i'$ as seen in figure \ref{fig:rectLower}(b) and 
$l$ be the vertical line passing through $a_1$. We have 2 cases : \\

\noindent Case 1 : If $a_1$ is the leftmost point in $P_i'$, then we remove $a_1$ from $P_i'$ and $R(x_i,a_1)$ from $X_i'$.
By the induction hypothesis, the lemma is true for the remaining $k-1$ points. On adding $a_1$ back, we see that the line 
$l$ contributes $k$ grid points to the next rectangle in $X_i'$, $R(x_i,a_2)$. This contribution of grid points by $l$ 
becomes $k-1$ for the next rectangle $R(x_i, a_3)$ and decreases by one as we move through the ordered set $X_i'$ and it 
is two for $R(x_i, a_k)$. Thus, the total number of points contributed by $l$ to $c$ is given by $\frac{k(k+1)}{2}-1$. The 
rectangle $R(x_i, a_1)$ also contributes $2k+2$ to $c$. Thus, $c \geq \frac{(k-1)^3}{6} + \frac{k(k+1)}{2} + (2k + 1) \geq 
\frac{k^3}{6}$. Thus, the statement is true for $m_i'=k$. \\ 

\noindent Case 2 : If $a_1$ is not the leftmost point, then we claim that $c$ does not increase when we make $a_1$ as the 
leftmost point by moving line $l$ to the left. To see this, refer figure \ref{fig:rectLower}(b) where the grid points on 
$l$ are shown as solid circles. Let $j$ be the number of points from $P_i'$ present to the left of $l$. When we make the point 
$a_1$ as the leftmost by moving $l$ to the left, we see that 

\begin{itemize}
\item The rectangles induced by $x_i$ and the points to the left of $l$ have an increase in the number of grid points, which 
is contributed by $l$. Thus, $c$ increases by $t \leq k+(k-1)+...+(k-j+1)=\frac{j(2k+1-j)}{2}$.
\item $R(x_i,a_1)$ loses $d = (j+2)(k+1)-2(k+1)=j(k+1)$ points. Thus, $c$ decreases by $d$.
\item The number of grid points in the rectangles induced by $x_i$ and the points to the right of $l$ remains the same.
\end{itemize}
By a simple calculation we can see that $d \geq t$. Thus, when $a_1$ is moved to the left, $c$ does not increase. As $a_1$ is 
now the leftmost point, we can apply case 1 and show that the lemma is true for $m_i'= k$.
\end{proof}

\begin{theorem}
\label{thm:secondRectangles}
Let $P$ be a point set of size $n$ in $\reals^2$ and let $\mathcal{S}$ be a set of induced rectangles of size $m$. If $m =
\Omega(n^{\frac{4}{3}})$, then there exists a point $p \in G$ which is present in at least $\frac{m^3}{24n^4}$ rectangles of 
$\mathcal{S}$.
\end{theorem}

\begin{proof}
The summation of the number of grid points present in the rectangles in $X_i$ is given by $\sum_{r \in X_i} J_r 
= \sum_{r \in X_i'} J_r + \sum_{r \in X_i''} J_r$. Using the lower bound from lemma \ref{lemma:rectLower} we have, 
$\sum_{r \in X_i} J_r \geq \frac{(m_i')^3 + (m_i'')^3}{6}$.

Since $\mathcal{S}$ is partitioned into the sets $X_i$, the summation of the number of grid points present in the rectangles 
in $\mathcal{S}$ is given by
\begin{equation*}
\sum_{r \in \mathcal{S}} J_r = \sum_{i=1}^n \sum_{r \in X_i} J_r \geq \left( \displaystyle \sum_{i=1}^n (m_i')^3 + \sum_{i=1}^n 
                                  (m_i'')^3 \right) /6
\end{equation*}
Using H\"older's inequality in $\reals^n$ (generalization of the Cauchy-Schwartz inequality), we have $\sum_{i=1}^n (m_i')^3 
\geq \frac{(m')^3}{n^2}$. Thus, we get $\sum_{r \in \mathcal{S}} J_r \geq \frac{(m')^3 + (m'')^3}{6n^2}$. This sum is minimized 
when $m' = m'' = \frac{m}{2}$ and thus, $\sum_{r \in \mathcal{S}} J_r \geq \frac{m^3}{24n^2}$.

Let $I_g$ be the number of rectangles of $\mathcal{S}$ containing the grid point $g \in G$. Now, by double counting, we have
\begin{equation*}
\begin{split}
\displaystyle \sum_{g \in G} I_g = \sum_{r \in \mathcal{S}} J_r \implies  \sum_{g \in G} I_g \geq \frac{m^3}{24n^2} 
\end{split}
\end{equation*}

By pigeonhole principle, there exists a grid point $p \in G$ which is present in at least $\frac{m^3}{24n^4}$ rectangles in 
$\mathcal{S}$.
\end{proof}

\subsection{Second selection lemma for other objects in $\reals^2$}
In this section, we look at the second selection lemma for objects like skylines and downward facing equilateral 
triangles.

Smorodinsky and Sharir~\cite{SHA04} proved tight bounds for the second selection lemma for disks. They used the planarity of the Delaunay graph (w.r.t circles) to prove that there 
exists a point $p \in P$ which is present in at least $\Omega(\frac{m^2}{n^2}) $ disks of $D$. It is not hard to see that this result applies for all objects whose Delaunay graph is planar.

\subsubsection{Skylines}
Let $\mathcal{K'} \subseteq \mathcal{K}$ be a set of $m$ skylines induced by $P$. It can be easily seen that the Delaunay graph
w.r.t skylines is planar. We can directly use the result in~\cite{SHA04} to get upper and lower bounds on the second selection 
lemma for induced skylines.
\begin{lemma}
There exists a point $p \in P$, which is present in $\Omega(\frac{m^2}{n^2})$ skylines induced by $P$. This bound is 
asymptotically tight.
\end{lemma}

\subsubsection{Downward facing equilateral triangles}
Let $\mathcal{T}$ be the set of all downward facing equilateral triangles or down-triangles induced by $P$. Such a triangle is 
induced by two points where the side parallel to the $x$-axis passes through one of the points and the corner opposite to this 
side lies below it. The other inducing point is present on one of the other 2 sides. Let $\mathcal{T'} \subseteq \mathcal{T}$ be a 
set of $m$ induced down-triangles.~\cite{BAMS12} proved that the Delaunay graph w.r.t to down-triangles is planar. Thus, we 
can apply the result in~\cite{SHA04} directly to get upper and lower bounds.
\begin{lemma}
There exists a point $P \in P$, which is present in $\Omega(\frac{m^2}{n^2})$ down-triangles induced by $P$. This bound is 
asymptotically tight.
\end{lemma}

\section{Conclusion}

In this paper, we have studied selection lemma type questions for various geometric objects.
We have proved exact results for both the strong and weak variants of the first selection lemma for 
axis-parallel rectangles and special subclasses like quadrants and slabs.
For the weak variant of the first selection lemma for axis-parallel boxes in $\reals^d$ though, there is a wide gap
between our lower bounds ($\frac{n^2}{2^{(2^d-1)}}$) and our upper bounds ($\frac{n^2}{2^{d+1}}$), which needs to be tightened.
We have shown non trivial bounds for the weak variant of first selection lemma for induced hyperspheres. Finding the exact constant
is an interesting open problem. Another open problem is to find non-trivial bounds for the strong variant of first selection lemma for boxes
and hyperspheres in higher dimensions.


For the second selection lemma for axis-parallel rectangles, we have proved a lower bound of $\frac{m^3}{24n^4}$ which is
a better bound than~\cite{SHA04}, when $m = \Omega(\frac{n^2}{\log^2n})$. An
interesting open problem, as mentioned in \cite{SHA04}, is to tighten the polylogarithmic gap between these lower and upper bounds.

\section*{Acknowledgements}
We would like to thank Neeldhara Misra and Saurabh Ray for helpful discussions.

\bibliographystyle{abbrv}

\end{document}